\documentclass[11pt,a4paper,reqno]{amsart}%
\usepackage[latin1]{inputenc}
\usepackage{mathrsfs}
\usepackage{dsfont}
\usepackage{hyperref}
\usepackage{amsmath}
\usepackage{amssymb}
\usepackage{amsthm}
\usepackage{amsfonts}
\usepackage{amstext}
\usepackage{amsopn}
\usepackage{amsxtra}
\usepackage{mathrsfs}
\usepackage{dsfont}
\usepackage{esint}
\usepackage{pst-all}
\usepackage{graphicx}
\theoremstyle{plain}
\newtheorem{theorem}{Theorem}[section]
\newtheorem{lemma}[theorem]{Lemma}
\newtheorem{corollary}[theorem]{Corollary}

\theoremstyle{definition}

\theoremstyle{remark}

\numberwithin{equation}{section}

\newcommand{\ii}{\infty}
\newcommand\R{{\ensuremath {\mathbb R} }}
\newcommand\C{{\ensuremath {\mathbb C} }}

\newcommand\1{{\ensuremath {\mathds 1} }}
\renewcommand\phi{\varphi}
\newcommand{\gH}{\mathfrak{H}}
\newcommand{\bH}{\mathbb{H}}

\newcommand{\wto}{\rightharpoonup}

\newcommand{\cM}{\mathcal{M}}

\newcommand{\cE}{\mathcal{E}}

\newcommand{\cF}{\mathcal{F}}
\newcommand{\cN}{\mathcal{N}}

\renewcommand{\epsilon}{\varepsilon}
\newcommand\pscal[1]{{\ensuremath{\left\langle #1 \right\rangle}}}
\newcommand{\norm}[1]{ \left| \! \left| #1 \right| \! \right| }
\newcommand{\tr}{{\rm Tr}\,}

\renewcommand{\ge}{\geqslant}

\renewcommand{\geq}{\geqslant}
\renewcommand{\leq}{\leqslant}

\renewcommand{\tilde}{\widetilde}

\newcommand{\nn}{\nonumber}
\newcommand{\eH}{\ensuremath{e_{\text{\textnormal{GP}}}}}

\newcommand{\dGamma}{{\rm d}\Gamma}
\title{Mean-Field limit of Bose systems: rigorous results}

\author[M. Lewin]{Mathieu Lewin}
\address{CNRS \& Universit\'e Paris-Dauphine, CEREMADE (UMR 7534), Place de Lattre de Tassigny, F-75775 Paris Cedex 16, France} 
\email{mathieu.lewin@math.cnrs.fr}

\date{\today}

\begin{document}

\begin{abstract}
We review recent results about the derivation of the Gross-Pitaevskii equation and of the Bogoliubov excitation spectrum, starting from many-body quantum mechanics. We focus on the mean-field regime, where the interaction is multiplied by a coupling constant of order $1/N$ where $N$ is the number of particles in the system. 

\medskip

\noindent\scriptsize Proceedings from the International Congress of Mathematical Physics at Santiago de Chile, July 2015.
\end{abstract}

\maketitle

\section{The Gross-Pitaevskii equation: a success story}

\subsection{The curse of dimensionality} 
Due to its high dimension, the many-body Schrödinger equation is impossible to solve numerically at a high precision for most physical systems of interest. It is therefore important to rely on simpler approximations, that are both precise enough and suitable to numerical investigation. One of the most famous is the \emph{mean-field model}, which consists in assuming that the particles are independent but evolve in an effective, self-consistent, one-body potential that replaces the many-particle interaction. Thereby, the high dimensional \emph{linear} many-body Schrödinger equation is replaced by a more tractable \emph{nonlinear} one-body equation. The paradigm is that the nonlinearity of the effective model should be able to reproduce the physical subtleties of the interaction in the real model. Introduced by Curie and Weiss to describe phase transitions in the classical Ising model, the mean-field method is now extremely popular in many areas of physics, and has even spread to other fields like biology and social sciences. 

In mathematical terms, the mean-field potential can be understood with the law of large numbers. Indeed, if the particles are independent and identically distributed according to a probability density $\rho$, then the interactions between the $j$th particle and all the others behave for a large number $N$ of particles as
\begin{equation}
\frac{1}{N-1}\sum_{\substack{k=1\\ k\neq j}}^Nw(x_j-x_k)\simeq \int_{\R^d} w(x_j-y)\,\rho(y)\,dy=w\ast\rho(x_j),
\label{eq:law_large_numbers}
\end{equation}
which is by definition the mean-field potential.

\subsection{A nonlinear equation}
One of the most famous mean-field model has been introduced by Gross~\cite{Gross-58} and Pitaevskii~\cite{Pitaevskii-61} for Bose gases, and it is similar to the Ginzburg-Landau theory of superconductivity. Historically designed to describe quantized vortices in superfluid Helium (in which it applies to only a small fraction of the particles), the Gross-Pitaevskii (GP) equation is now the main tool to understand the Bose-Einstein condensates which have been produced in the laboratory with ultracold trapped gases, since the 90s~\cite{CorWie-95,Ketterle-95,FetFoo-12}. 

For $N$ bosons in $\R^d$ with $d=1,2,3$, the GP equation takes the form
\begin{equation}
hu_0+(N-1)\big(|u_0|^2\ast w\big)u_0=\epsilon_0\,u_0,\qquad \int_{\R^d}|u_0|^2=1,
\label{eq:GP}
\end{equation}
where $u_0$ is the quantum state common to the condensed (independent) particles. This equation can be obtained by minimizing the Gross-Pitaevskii energy
\begin{equation}
\cE(u)=\int_{\R^d}u(x)^* (hu)(x)\,dx+\frac{N-1}{2}\int_{\R^d}\int_{\R^d}w(x-y)|u(x)|^2|u(y)|^2\,dx\,dy.
\label{eq:GP_energy_intro}
\end{equation}
In the literature, the name ``Hartree'' is often used in place of ``Gross-Pitaevskii'' when $w$ is a smooth function. When $w$ is proportional to a Dirac delta, one often uses the acronym NLS for ``Nonlinear Schrödinger''.

In Equation~\eqref{eq:GP}, the mean-field potential is $(N-1)\big(|u_0|^2\ast w\big)$ due to~\eqref{eq:law_large_numbers}, and $h$ is the noninteracting one-body Hamiltonian. We take
$$h=\big(-i\nabla+A(x)\big)^2+V(x)$$
where $A$ plays the role of a magnetic potential but usually arises from a different physical origin (for instance, for rotating trapped gases, in the rotating frame $A$ would describe the Coriolis force). The potential $V$ includes both the external potential applied to the system (for instance a harmonic potential in a harmonic trap) as well as the centrifugal force for a rotating system. If the particles have an internal degree of freedom taking $n$ possible values, then we allow $V(x)$ to be a $n\times n$ hermitian matrix. More complicated models can be considered without significantly changing the results presented in this paper. We want to keep $A$ and $V$ as general as possible and we will only distinguish two situations: the case of \emph{trapped systems} where $h$ is assumed to have a compact resolvent, and \emph{untrapped or locally trapped systems} where $A$ and $V$ tend to zero at infinity. The interaction potential $w$ will always be assumed to vanish at infinity in a suitable sense, and to be an even function.

The equation~\eqref{eq:GP} has been used with impressive success to describe Bose-Einstein condensates. For 3D trapped dilute gases of alkaline atoms, 
$$w(x)=8\pi \, a\, \delta(x)$$
where $a$ is the scattering length of the interaction, leading to the nonlinear Schrödinger equation with a cubic nonlinearity
\begin{equation}
hu_0+8\pi a (N-1)|u_0|^2u_0=\epsilon_0\,u_0. 
 \label{eq:NLS}
\end{equation}
In dipolar gases the long-range dipole-dipole interaction must be added and 
$$w(x)=8\pi \, a\, \delta(x)+\kappa\frac{1-3\cos^2\theta}{|x|^3}$$
where $\theta$ is the angle of $x$ relative to the fixed dipole direction. These two examples lead to the simplest equations, but general $w$'s are also of physical interest. It is possible to experimentally produce tunable long-range interactions in a dilute gase by making it interact with a cavity~\cite{Mottl-12}.

The nonlinear character of the Gross-Pitaevskii equation~\eqref{eq:GP} has been shown to properly reproduce some intriguing properties of Bose-Einstein condensates, through the breaking of symmetries. A famous example is that of the vortices that appear in rotating gases, which are perfectly well described by GP theory (Figure~\ref{fig:vortices})~\cite{Yngvason-14,Aftalion-06,Fetter-09}.

\begin{figure}[h]
\centering
\includegraphics[height=5cm]{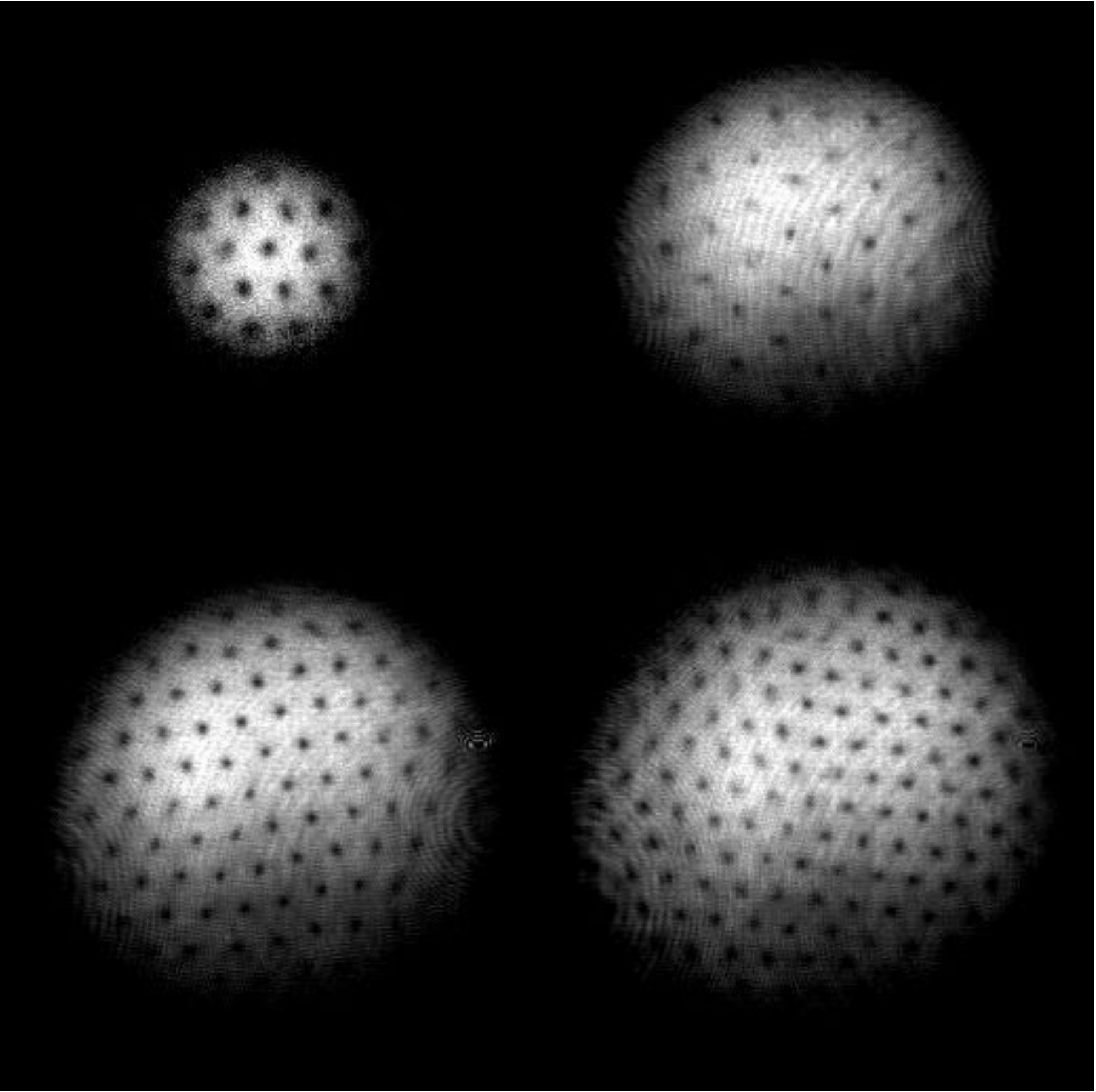}
\includegraphics[height=5.4cm]{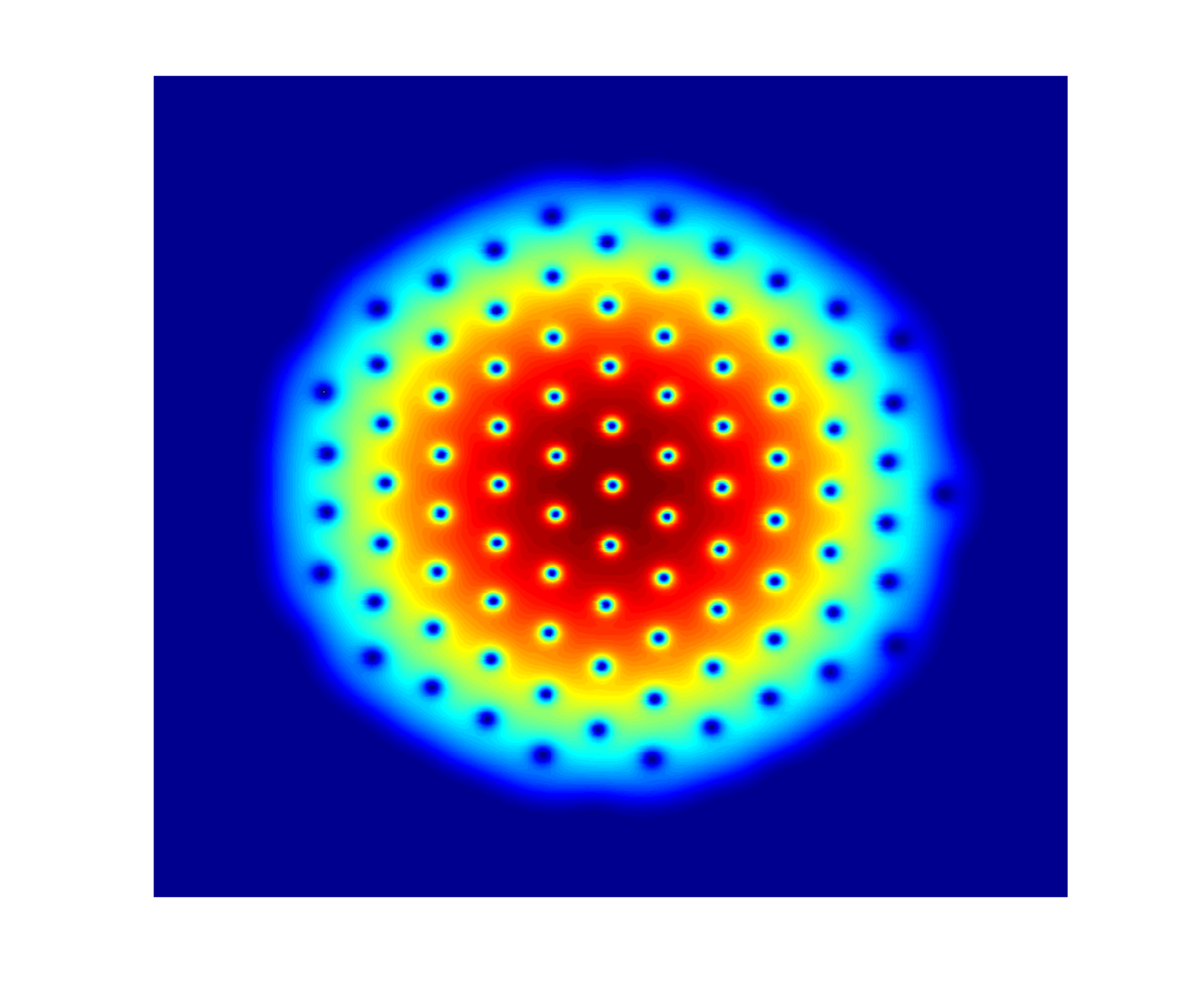}
\caption{\small \emph{Left:} Experimental pictures of fast rotating Bose-Einstein condensates, taken from~\cite{Ketterle-01}, \copyright\; AAAS.
  \emph{Right:} Numerical calculation of the Gross-Pitaevskii solution with the software GPELab~\cite{AntDub-14a,AntDub-14b} in the corresponding regime. 
\label{fig:vortices}}
\end{figure}

\subsection{Bogoliubov excitation spectrum}
Another important feature of the GP equation is its ability to explain the superfluid character of these Bose gases, which is understood in terms of the excitation spectrum~\cite{Landau-41,Bogoliubov-47}. Bogoliubov's theory predicts that the excited energies will be given by the corresponding Bogoliubov Hamiltonian
\begin{multline}
\bH_{u_0}=\int a^\dagger(x)\big(h+(N-1)\,|u_0|^2\ast w-\epsilon_0\big)a(x)\,dx\\+(N-1)\int\int u_0(x)\overline{u_0(y)}w(x-y)a^\dagger(x)a(y)\,dx\,dy\\
+\frac{N-1}2\int\int u_0(x)u_0(y)w(x-y)a^\dagger(x)a^\dagger(y)\,dx\,dy+c.c. 
\label{eq:Bogoliubov}
\end{multline}
which is the second quantization of the Hessian of the Gross-Pitaevskii energy at the solution $u_0$. This second-quantized Hamiltonian is defined on the bosonic Fock space with the mode $u_0$ removed. For spinless particles in a box with periodic boundary conditions and when $u_0\equiv1$, the spectrum of the Hamiltonian $\bH_{u_0}$ is given in terms of the elementary excitations 
$$\sqrt{|k|^4+2(2\pi)^{d/2}(N-1)\widehat{w}(k)|k|^2}.$$
When $w\propto \delta$, the effective dispersion relation is increasing and linear at small momentum, a fact that has been confirmed in experiments for alkaline cold gases in~\cite{VogXuRamAboKet-02,SteOzeKatDav-02}. For other interactions the excitation spectrum can display the famous roton local minimum (Figure~\ref{fig:phonon-roton}). This has been recently experimentally confirmed for dilute gases with a spin-orbit term in~\cite{JiZhaXuWu-15} and for long-range interactions mediated through a cavity in~\cite{Mottl-12}. For dipolar gases, the roton spectrum is so far only a theoretical prediction~\cite{SanShlLew-03}.

\begin{figure}[h]
\centering
\includegraphics[width=4cm]{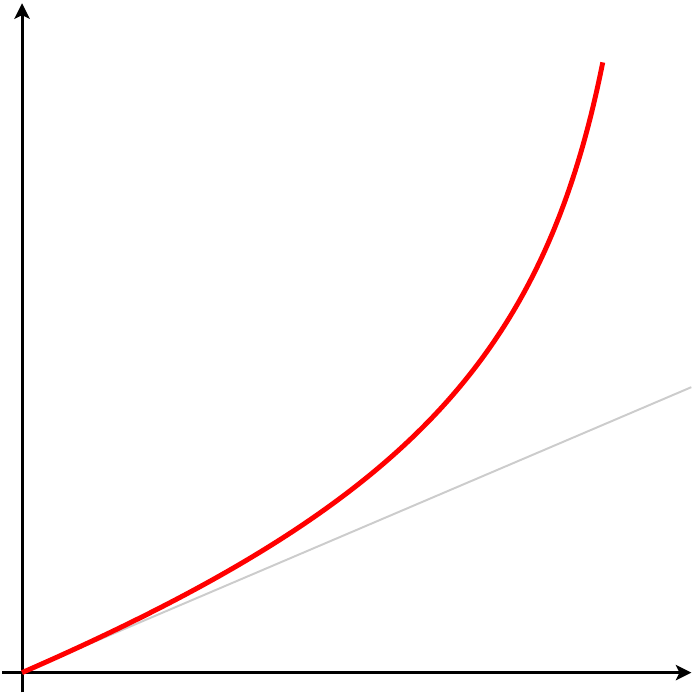}\hspace{1cm}\includegraphics[width=4cm]{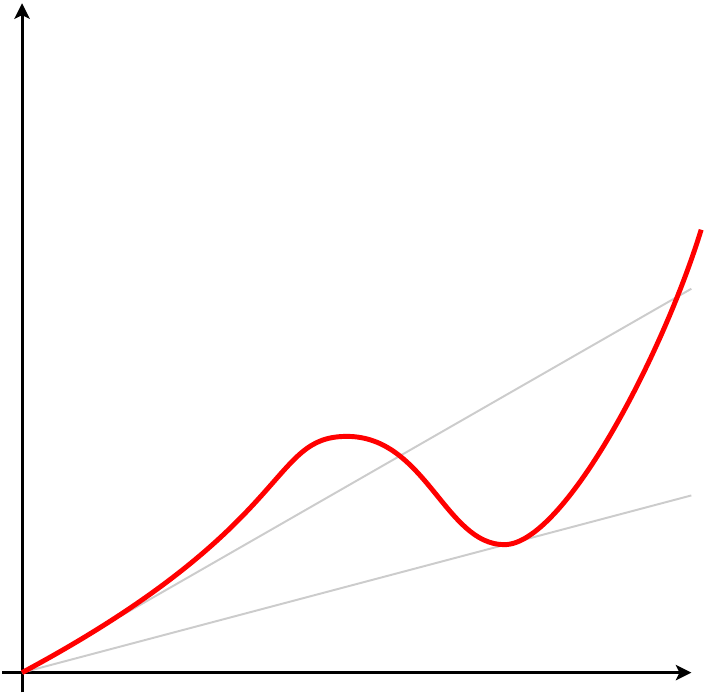}
\caption{\small Pure phonon excitation spectrum (\emph{left}) when $w\propto \delta$, as opposed to the phonon-maxon-roton excitation spectrum (\emph{right}) which can be obtained for a more general $w$ and has been experimentally confirmed in some cases.
\label{fig:phonon-roton}}
\end{figure}

\subsection{Validity of Gross-Pitaevskii and Bogoliubov}

The successes of the Gross-Pitaevskii equation (and of the associated Bogoliubov predictions) have stimulated many works in mathematical physics. Many authors have worked on deriving properties of solutions to the GP equation, with a particular focus on the vortices. A different route consists in trying to justify the GP approximation starting from first principles (that is, the many-body Schrödinger equation) and the main goal of this proceeding is to review recent advances in this direction.

The proof of Bose-Einstein condensation requires to understand how independence can arise in an interacting system. Clearly the interactions will have to be weak (but not too much such as to remain in the final effective GP equation). There are (at least) two ways this could happen. The first is when the interactions are \textbf{rare}. That is, the particles in the system meet very rarely and, when they do, they interact with a $w$ of order one. This situation corresponds to the \textbf{dilute regime} which is appropriate for the ultracold Bose gases produced in the lab. Another regime is when the interactions are \textbf{small in amplitude}, that is, $w$ is small itself. This is the typical case for the law of large number~\eqref{eq:law_large_numbers} to hold, since the factor $1/(N-1)$ there can be identified to a small coupling constant $\lambda=1/(N-1)$. In order to end up with an interesting model, we then need many collisions and this now corresponds to a \textbf{high density regime} where the particles meet very often but interact only a little bit each time.

\begin{table}[h]
\begin{tabular}{|l|l|}
\hline
low density & rare interactions of intensity 1\\
\hline
high density & frequent interactions of small intensity $\lambda\sim 1/N$\\
\hline
\end{tabular}

\bigskip

\caption{\small Two completely opposite physical situations in which the GP equation will be valid, for different reasons. In the mathematical physics literature, the low density regime is often called the \emph{Gross-Pitaevskii limit} whereas the high density regime is often called the \emph{mean-field limit}.}
\end{table}

These two regimes are very different but they lead to somewhat similar effects. This similarity is both very useful and confusing. We will explain later the exact difference, but the philosophy is that everything that holds in the high density regime with a potential $w/N$ holds similarly in the low density regime with the potential $8\pi a\delta$ where $a$ is the scattering length of $w$. The occurrence of the scattering length is due to a two-particle scattering process that induces short-distance strong correlations and soils the hoped-for independence between the particles. This phenomenon is very difficult to capture mathematically~\cite{LieSeiSolYng-05}.
In this paper we will review recent (and older) results concerning the second regime, which is much simpler mathematically. As we will see, new tools developed in simple cases can sometimes be useful in more complicated situations.

\section{The mean-field model}

We assume that the particles interact with a potential $\lambda w$, leading to the many-particle Hamiltonian
\begin{equation}
 \boxed{H_N=\sum_{j=1}^N h_{x_j}+\lambda\sum_{1\leq j<k\leq N}w(x_j-x_k),} 
 \label{eq:H_N}
\end{equation}
acting on the bosonic space $\bigotimes_s^NL^2(\R^d)$. We are interested in the regime of a large number of particles $N\gg1$ with small coupling constant $\lambda\ll1$. In view of the discussion in the previous section and, in particular, of the law of large numbers~\eqref{eq:law_large_numbers}, it is natural to assume that $\lambda\sim1/N$. This makes the two terms in $H_N$ of the same order $N$. Without loss of generality and in order to simplify some expressions, we will take
$$\boxed{\lambda=\frac{1}{N-1}.}$$
For this model we will prove the validity of the Gross-Pitaevskii and Bogoliubov theories in the limit $N\to\ii$, at least for the particles that do not escape to infinity. These particles will be essentially packed in finite regions of space, creating a local density of order $N\gg1$. Since the GP equation will always be right for $H_N$ in the limit of large $N$, it is customary to call it the \emph{mean-field model}.

The mean-field model is not a physical model since there is no reason to believe that the interaction would be multiplied by $1/N$ in a real system. However it can be used for dilute gases with long-range tunable interactions~\cite{Mottl-12}, as well as for some other systems that can be recast in the form~\eqref{eq:H_N} in an appropriate regime. This includes for instance bosonic atoms~\cite{BenLie-83,Solovej-90,Bach-91,BacLewLieSie-93,Kiessling-12} when the number of electrons is proportional to the nuclear charge, and stars~\cite{LieYau-87}. The Hamiltonian $H_N$ is an interesting mathematical prototype for Bose-Einstein condensation, which can teach us new things that could be useful in more physical situations.

\medskip

The most natural way to understand the occurrence of the GP equation is to \emph{assume} that the particles are \emph{all exactly independent}. As we will explain, this is certainly wrong as such, since a small proportion of the particles will always behave differently. But this simple-minded argument gives the right energy on first order. Mathematically, this corresponds to choosing 
$$\Psi(x_1,...,x_N)=u(x_1)\cdots u(x_N):=u^{\otimes N}(x_1,...,x_N)$$
for some normalized function $u\in L^2(\R^d)$. Computing the quantum energy using $\lambda=1/(N-1)$ gives 
$\pscal{u^{\otimes N},H_Nu^{\otimes N}}=N\,\cE(u)$
where 
\begin{equation}
\cE(u)=\pscal{u,hu}+\frac{1}{2}\int_{\R^d}\int_{\R^d}w(x-y)|u(x)|^2|u(y)|^2\,dx\,dy.
\label{eq:GP_energy}
\end{equation}
Minimizers of $\cE$ satisfy the GP equation
\begin{equation}
hu_0+\big(|u_0|^2\ast w\big)u_0=\epsilon_0\,u_0,
\label{eq:GP2}
\end{equation}
where $\epsilon_0$ is a Lagrange multiplier due to the normalization constraint. In general there is no reason to expect that minimizers will be unique and, indeed, uniqueness should definitely \emph{not} hold when there are vortices. Uniqueness can be broken because of the one-body Hamiltonian $h$ (e.g. for rotating gases), or due to the interaction if it has an attractive part.

\medskip

To avoid any confusion, let us now compare the mean-field model to the dilute limit in 3D. Take for simplicity $A=V=0$ and put $N$ particles interacting with a potential $w$ in a box of volume $\ell^3$. For a repulsive potential the particles will be spread over the whole box at an average distance $N^{-1/3}\ell$ to one another. The dilute limit corresponds to taking $N\to\ii$ at the same time as $N^{-1/3}\ell\to\ii$ to ensure that they meet rarely (and then $\rho=N/\ell^3\to0$). By scaling we can recast everything in a fixed box of side length $1$, leading to the Hamiltonian
$$\ell^{-2}\sum_{j=1}^N -\Delta_{x_j}+\sum_{1\leq j<k\leq N}w\big(\ell(x_j-x_k)\big).$$
This Hamiltonian can be written in a form similar to~\eqref{eq:H_N} if we let $\ell=N$, $w_N(x)=N^3w(Nx)$ and multiply everything by $N^2$: 
$$\sum_{j=1}^N -\Delta_{x_j}+\frac{1}{N}\sum_{1\leq j<k\leq N}w_N(x_j-x_k).$$
This is a mean-field model in which the interaction is scaled with $N$ and satisfies $w_N\wto (\int_{\R^3}w)\delta$. This could suggest that one would end up with the GP equation with interaction $(\int_{\R^3}w)\delta$, which happens to be somewhat wrong. There are short-distance correlations but their sole effect is to renormalize the constant $\int_{\R^3}w$ to $8\pi a$, where $a$ is the scattering length of $w$~\cite{LieYng-98,LieSeiYng-00,LieSeiSolYng-05,LieSei-06,NamRouSei-15}. It is because the Hamiltonian has a form similar to the mean-field model~\eqref{eq:H_N} that some tools developed for the latter can sometimes be useful in the dilute limit.

In an intermediate model where $w_N$ is scaled at a slower rate, e.g. $w_N(x)=N^{3\beta}w(N^\beta x)$ with $0<\beta<1$, the limit will indeed involve the constant $\int_{\R^3}w$~\cite{LewNamRou-14c} instead of $8\pi a$. Here  $\beta = 1/3$ is the dividing line between ``rare but strong" and ``frequent but weak" interactions. The proof for $1/3<\beta<1$ is therefore much more delicate, since the typical interaction length is now much smaller than the average distance between the particles.


\section{Convergence of the ground state energy}\label{sec:energy}

In this section we are going to explain that the ground state energy is, to leading order, given by states with independent particles as above. To this end, we introduce the many-particle and GP ground state energies
$$E(N)=\inf\sigma_{\bigotimes_s^NL^2(\R^d)}(H_N),\qquad \eH:=\inf_{\int_{\R^d}|u|^2=1}\cE(u).$$
From the variational principle it is clear that $E(N)\leq N\,\eH$.

We need technical assumptions on the potentials $A$, $V$ and $w$ to make everything meaningful. We distinguish two situations
\begin{description}
 \item[$\bullet$ confined case] $h$ is bounded from below and has a compact resolvent, and $|w(x-y)|$ is $h_x+h_y$ form-bounded with relative bound $<1$, on the two-particle space $L^2(\R^d)\otimes_s L^2(\R^d)$;
 \item[$\bullet$ unconfined or locally-confined case] the real-valued functions $|V|$, $|A|^2$ and $w$ are in $L^p(\R^d)+L^\ii_\epsilon(\R^d)$ with $p\geq\max(2,d/2)$ ($p>2$ in dimension $d=4$), hence are $-\Delta$--compact. 
\end{description}
These conditions are not optimal and can be weakened in several ways but we will not discuss this here. The most important is that we make no assumption on the sign of $w$ (nor on its Fourier transform $\widehat{w}$). The interaction can be repulsive, attractive, or both. The following was proved in~\cite{LewNamRou-14}.

\begin{theorem}[Convergence of energy~\cite{LewNamRou-14}]\label{thm:energy}
Under the previous assumptions, we have
$$\boxed{\lim_{N\to\ii}\frac{E(N)}{N}=\eH.}$$
\end{theorem}

Results similar to Theorem~\ref{thm:energy} have been shown in many particular situations, but Theorem~\ref{thm:energy} is, to our knowledge, the first generic result. 
Previous works dealt with the Lieb-Liniger model~\cite{LieLin-63,SeiYngZag-12}, bosonic atoms~\cite{BenLie-83,Solovej-90,Bach-91,BacLewLieSie-93}, stars~\cite{LieYau-87}, and confined systems~\cite{FanSpoVer-80,VdBLewPul-88,RagWer-89,Werner-92}. 

We remark that, for confined systems, the result is exactly the same at positive temperature:
$$\lim_{N\to\ii}-\frac{1}{\beta N}\log\tr (e^{-\beta H_N})=\eH,$$
see~\cite[Thm. 3.6]{LewNamRou-14}. For bosons the entropy plays no role at the considered scale and it is necessary to take $T$ of order $N$ in order to see a depletion from the ground state energy. The limit is then much more delicate and involves the Gibbs measures associated with the GP nonlinear energy $\cE$~\cite{LewNamRou-15}.

Before discussing the convergence of states, we would like to give a simple proof of Theorem~\ref{thm:energy}, different from~\cite{LewNamRou-14}, that only works when $A\equiv0$ and $V$ is spin-independent. We start with the well-known case of a positive-definite interaction.

\subsubsection*{Proof for non-rotating gases with positive-definite interaction}
The proof of Theorem~\ref{thm:energy} is well known in the simplest situation in which
\begin{itemize}
 \item[(i)] $h$ is real and positive preserving, $\pscal{u,hu}\geq\pscal{|u|,h|u|}$, that is, $A\equiv0$ and $V$ is spin-independent;
 \item[(ii)] $w$ is positive-definite, that is, has a positive Fourier transform $\widehat{w}\geq0$.
\end{itemize}
These two properties can be used through the following two lemmas.
\begin{lemma}[Hoffmann-Ostenhof inequality~\cite{Hof-77}]
If $h$ is real and positive-preserving, then for every bosonic many particle state $\Psi_N$ we have
\begin{equation}
\Big\langle \Psi_N,\sum_{j=1}^Nh_{x_j}\Psi_N\Big\rangle\geq N\pscal{\sqrt{\rho_{\Psi_N}},h\sqrt{\rho_{\Psi_N}}}
\label{eq:Hoffmann-Ostenhof}
\end{equation}
with the one-particle density
$$\rho_{\Psi_N}(x)=\int_{\R^d}\cdots\int_{\R^d}|\Psi_N(x,x_2,...,x_N)|^2\,dx_2\cdots dx_N.$$
\end{lemma}
\begin{proof}
Written in a basis of natural orbitals $u_k$ and occupation numbers $n_k$, the proof is just the remark that 
\begin{align*}
 \Big\langle \Psi_N,\sum_{j=1}^Nh_{x_j}\Psi_N\Big\rangle =N\sum_{k\geq1}n_k\pscal{u_k,hu_k}&\geq N\sum_{k\geq1}n_k\pscal{|u_k|,h|u_k|}\\
 &\geq N\pscal{\sqrt{\sum_{k\geq1}n_k|u_k|^2},h\sqrt{\sum_{k\geq1}n_k|u_k|^2}}
\end{align*}
where in the last line we have inductively applied the inequality 
$$\pscal{u_1,hu_1}+\pscal{u_2,hu_2}=\pscal{u_1+iu_2,h(u_1+iu_2)}\geq \pscal{\sqrt{u_1^2+u_2^2},h\sqrt{u_1^2+u_2^2}}$$ for real functions $u_1,u_2$.
\end{proof}

\begin{lemma}[Estimating the two-body interaction by a one-body term]
If $\widehat{w}\geq0$ is in $L^1(\R^d)$, then for all $\eta\in L^1(\R^d)$
\begin{multline}
\sum_{1\leq j<k\leq N}w(x_j-x_k)\\ \geq \sum_{j=1}^N\eta\ast w(x_j)-\frac12\int_{\R^d}\int_{\R^d}w(x-y)\eta(x)\eta(y)\,dx\,dy-\frac{N}{2}w(0).
\label{eq:lower_bound_w}
\end{multline}
\end{lemma}

\begin{proof}
With $f=\sum_{j=1}^N\delta_{x_j}-\eta$, expand $\int_{\R^{d}}\int_{\R^{d}}w(x-y)f(x)f(y)\,dx\,dy=(2\pi)^{d/2}\int_{\R^d}\widehat{w}(k)|\widehat{f}(k)|^2\,dk\geq0$. 
\end{proof}

Taking $\eta=N\rho_{\Psi_N}$ and using the Hoffmann-Ostenhof inequality~\eqref{eq:Hoffmann-Ostenhof}, we obtain the lower bound
\begin{equation}
\pscal{\Psi_N,H_N\Psi_N}\geq N\cE(\sqrt{\rho_{\Psi_N}})-\frac{N\,w(0)}{2(N-1)}\geq N\,\eH-\frac{N\,w(0)}{2(N-1)}. 
 \label{eq:lower_bound_GP_density}
\end{equation}
Minimizing over $\Psi_N$ and recalling the upper bound $E(N)\leq N\,\eH$ proves the final estimate
$$\eH-\frac{w(0)}{2(N-1)}\leq\frac{E(N)}{N}\leq \eH,$$
which clearly ends the proof of Theorem~\ref{thm:energy} when $0\leq \widehat{w}\in L^1(\R^d)$. If $\widehat{w}$ is positive but not integrable (e.g. for Coulomb), the proof can be done by an approximation argument.\qed

\subsubsection*{Proof for non-rotating gases with an arbitrary interaction}
We now sketch an unpublished proof of Theorem~\ref{thm:energy} for an arbitrary interaction $w$, still under the assumption (i) that $h$ is real and positive-preserving. 
The idea, inspired by~\cite{LevyLeblond-69,LieYau-87}, is to use auxiliary classical particles that repel each other, in order to model the attractive part of the interaction. 

For simplicity we consider $2N$ particles that we split in two groups of $N$. The positions of the $N$ first will be denoted by $x_1,...,x_N$ whereas those of the others will be denoted by $y_1=x_{N+1},...,y_N=x_{2N}$. Of course, the separation is completely artificial and in reality the $2N$ particles are indistinguishable. Next we pick a $2N$-particle state $\Psi_{2N}$ and use its bosonic symmetry in the $2N$ variables to rewrite
$$\frac1{2N}\Big\langle \Psi_{2N},\sum_{j=1}^{2N}h_{x_j}\Psi_{2N}\Big\rangle=\frac1{N}\Big\langle \Psi_{2N},\sum_{j=1}^{N}h_{x_j}\Psi_{2N}\Big\rangle.$$
In a similar fashion, we decompose $w=w_1-w_2$ where $\widehat{w_1}=(\widehat{w})_+\geq0$ and $\widehat{w_2}=(\widehat{w})_-\geq0$ and write the repulsive part using only the $x_j$'s
\begin{multline*}
\frac1{2N(2N-1)}\Big\langle \Psi_{2N},\sum_{1\leq j<k\leq 2N}w_1(x_j-x_k)\Psi_{2N}\Big\rangle\\=\frac1{N(N-1)}\Big\langle \Psi_{2N},\sum_{1\leq j<k\leq N}w_1(x_j-x_k)\Psi_{2N}\Big\rangle.
\end{multline*}
On the other hand, we express the attractive part as the difference of two terms, involving respectively only the $y_\ell$'s and both species: 
\begin{multline*}
-\frac1{2N(2N-1)}\Big\langle \Psi_{2N},\sum_{1\leq j<k\leq 2N}w_2(x_j-x_k)\Psi_{2N}\Big\rangle\\=\frac1{N(N-1)}\Big\langle \Psi_{2N},\sum_{1\leq \ell<m\leq N}w_2(y_\ell-y_m)\Psi_{2N}\Big\rangle\\ - \frac1{N^2}\Big\langle \Psi_{2N},\sum_{j=1}^N\sum_{\ell=1}^Nw_2(x_j-y_\ell)\Psi_{2N}\Big\rangle.
\end{multline*}
This means that $\pscal{\Psi_{2N},H_{2N}\Psi_{2N}}/2N= \langle\Psi_{2N},\tilde H_N\Psi_{2N}\rangle/N$
with
\begin{multline*}
\tilde H_N=\sum_{j=1}^{N}h_{x_j}+\frac{1}{N-1}\sum_{1\leq j<k\leq 2N}w_1(x_j-x_k)+\frac{1}{N-1}\sum_{1\leq \ell<m\leq N}w_2(y_\ell-y_m)\\-\frac1N \sum_{j=1}^N\sum_{\ell=1}^Nw_2(x_j-y_\ell).
\end{multline*}
This Hamiltonian describes a system of $N$ quantum particles that repel through the potential $w_1/(N-1)$ and $N$ classical particles that repel through $w_2/(N-1)$, with an attraction $-w_2/N$ between the two species.

In order to bound $\tilde H_N$ from below, we first fix the positions $y_1,...,y_N$ of the particles in the second group and consider $\tilde H_N$ as an operator acting only over the $x_j$'s. Let $\Phi_N$ be any bosonic $N$-particle state in the $N$ first variables. Using~\eqref{eq:Hoffmann-Ostenhof} and~\eqref{eq:lower_bound_w} for the repulsive potential $w_1$ as in the previous proof, we obtain
\begin{multline*}
\frac{\pscal{\Phi_N,\tilde H_N\Phi_N}}N\geq \pscal{\sqrt{\rho_{\Phi_N}},h\sqrt{\rho_{\Phi_N}}}+\frac12 \int_{\R^d}\int_{\R^d}\rho_{\Phi_N}(x)\rho_{\Phi_N}(y)w_1(x-y)\,dx\,dy\\
-\frac{w_1(0)}{2(N-1)}+\frac{1}{N(N-1)}\sum_{1\leq \ell<m\leq N}w_2(y_\ell-y_m)
-\frac1N\sum_{\ell=1}^N \rho_{\Phi_N}\ast w_2(y_\ell).
\end{multline*}
Next we use again~\eqref{eq:lower_bound_w} for $w_2$ with $\eta=(N-1)\rho_{\Phi_N}$ and obtain
\begin{multline*}
\sum_{1\leq \ell<m\leq N}w_2(y_\ell-y_m)-(N-1)\sum_{\ell=1}^N \rho_{\Phi_N}\ast w_2(y_\ell)\\
\geq -\frac{(N-1)^2}{2}\int_{\R^d}\int_{\R^d}\rho_{\Phi_N}(x)\rho_{\Phi_N}(y)w_2(x-y)\,dx\,dy -\frac{Nw_2(0)}{2}.
\end{multline*}
Therefore, we have shown that
$$\frac{\pscal{\Phi_N,\tilde H_N\Phi_N}}N\geq \cE(\sqrt{\rho_{\Phi_N}})-\frac{w_1(0)+w_2(0)}{2(N-1)}\geq \eH-\frac{w_1(0)+w_2(0)}{2(N-1)}.$$
Since the right side is independent of the $y_\ell$'s, the bound 
$$\frac{\tilde H_N}N\geq \eH-\frac{w_1(0)+w_2(0)}{2(N-1)}$$
holds in the sense of operators in the $2N$-particle space.
Minimizing over $\Psi_{2N}$ gives our final estimate
$$\eH-\frac{w_1(0)+w_2(0)}{2(N-1)}\leq \frac{E(2N)}{2N}\leq \eH.$$
We have considered an even number of particles for simplicity, but the proof works the same if we split the system into two groups of $N$ and $N+1$ particles. Another possibility is to use that $N\mapsto E(N)/N$ is non-decreasing, which gives the final estimate 
$$\boxed{\eH-\frac{w_1(0)+w_2(0)}{N-3}\leq \frac{E(N)}{N}\leq \eH}$$
for $N\geq 4$. Note that $w_1(0)+w_2(0)=(2\pi)^{-d/2}\int_{\R^d}|\widehat{w}|$. Non-integrable $\widehat{w}$ can be handled using an approximation argument.\qed

\medskip

The previous two proofs rely deeply on the Hoffmann-Ostenhof inequality, which allows to use the square root of the density as a trial function for the GP energy $\cE$ in a lower bound. This argument can only be used when the GP minimizers are positive, that is, when there are no vortices.\footnote{There could be several GP minimizers for an attractive potential $w$, but they will all be positive when $h$ satisfies (i).} When $h$ does not satisfy the property (i), a completely different approach is necessary. 

We remark that the method applies to the translation-invariant case $A,V\equiv0$, since then $h=-\Delta$ satisfies the Hoffmann-Ostenhof inequality. As we will explain in the next section, this will be useful to deal with the particles that might escape the system, when the potentials $A$ and $V$ vanish at infinity. 

\section{Convergence of states and quantum de Finetti theorems}

We discuss here the link between the many-particle ground states $\Psi_N$ and the minimizers of the GP energy $\cE$, solving the GP equation~\eqref{eq:GP2}. We also indicate how to prove Theorem~\ref{thm:energy} for a general $h$. For a pedagogical presentation of these results from~\cite{LewNamRou-14,LewNamRou-15b}, we also refer to~\cite{Rougerie-15}.

As we will explain in Section~\ref{sec:Bogoliubov}, we cannot expect that $\Psi_N$ will be close in norm to a factorized state $u^{\otimes N}$. Instead, we use its $k$-particle density matrix, whose integral kernel is defined by
$$\Gamma_{\Psi_N}^{(k)}(x_1,...,x_k,y_1,...,y_k):=\int_{\R^{d(N-k)}}\Psi_N(x_1,...,x_k,Z)\overline{\Psi(y_1,...,y_k,Z)}\,dZ.$$
It is normalized to $\tr\Gamma_{\Psi_N}^{(k)}=1$, on the contrary to the usual convention: here macroscopically occupied states have an occupation number of order one. Remark that for a factorized state
$\Gamma_{u^{\otimes N}}^{(k)}=|u^{\otimes k}\rangle\langle u^{\otimes k}|.$
Our goal is, therefore, to prove that $\Gamma_{\Psi_N}^{(k)}$ converges to a density matrix in this form, with $u$ minimizing the GP energy. However, in the case of non-unique GP minimizers, several of them could be occupied in the limit, which is usually called \emph{fragmented Bose-Einstein condensation}. The best we can hope is that we obtain a convex combination of all the possible GP minimizers.

\begin{theorem}[Convergence of states~\cite{LewNamRou-14}]\label{thm:states}
Under the previous assumptions on $h$ and $w$, let $\Psi_N$ be any sequence of $N$-particle states such that 
$$\pscal{\Psi_N,H_N\Psi_N}=E(N)+o(N).$$
In the \textbf{confined case}, there exists a subsequence and a probability measure $\mu$ on the set 
$$\cM=\{\text{minimizers for $\eH$}\}$$ 
such that 
\begin{equation}
 \Gamma_{\Psi_{N_j}}^{(k)}\to \int_{\cM}|u^{\otimes k}\rangle\langle u^{\otimes k}|\,d\mu(u) 
 \label{eq:CV_PDM}
\end{equation}
strongly in the trace-class as $N_j\to\ii$, for all $k\geq1$. 
In the \textbf{unconfined or locally-confined case} the result is the same except that some particles could escape to infinity, hence we have to take
$$\cM=\{\text{weak limits of minimizing sequences for $\eH$}\}$$
and the limit in~\eqref{eq:CV_PDM} \emph{a priori} only holds weakly-$\ast$ in the trace-class.
\end{theorem}

The probability measure $\mu$ describes the fragmented Bose-Einstein condensation. For instance, if there are only two GP minimizers, then $\mu$ will give their relative occupations. The simplest case is when $\cM=\{u_0\}$, where there will always be complete Bose-Einstein condensation on $u_0$. In general we should probably think of $\mu$ as a probability over experiments, where only one GP minimizer $u$ is usually observed at a time. We remark that the result holds for any sequence $\Psi_N$ that provides the right ground state energy to leading order. It is indeed possible to construct sequences $\Psi_N$ that yield any probability $\mu$ on $\cM$ in the limit. If we take for $\Psi_N$ an exact ground state of $H_N$, then we might end up with a definite $\mu$, a question that is not addressed by the theorem.

In the unconfined or locally-confined case, the theorem only gives Bose-Einstein condensation for the particles that stay (due to the weak limits in the statement). All the information about the particles that escape to infinity is lost. In principle, all the particles could even escape, in which case $\cM=\{0\}$ and the result is essentially empty.\footnote{Nevertheless, one can apply a translation that follows one cluster, if it exists, and Theorem~\ref{thm:states} applied in this new reference frame gives Bose-Einstein condensation in the cluster.} Dealing with the possibility that some particles could fly apart is a delicate mathematical question which occupies a large part of~\cite{LewNamRou-14} and is our main contribution. The confined case was essentially known before~\cite{FanSpoVer-80,VdBLewPul-88,RagWer-89,Werner-92,LieSei-06}.

\medskip

The main tool for proving Theorems~\ref{thm:energy} and~\ref{thm:states} is the \emph{quantum de Finetti theorem}. The latter is an abstract result which says that, at the level of density matrices, only factorized states remain in the limit $N\to\ii$, for any sequence of bosonic states. This result is similar to the classical de Finetti-Hewitt-Savage theorem~\cite{DeFinetti-31,DeFinetti-37,Dynkin-53,HewSav-55,DiaFre-80} which states that the law of any sequence of exchangeable random variables is always a convex combination of laws of iid variables. The quantum analogue reads as follows.

\begin{theorem}[quantum de Finetti~\cite{Stormer-69,HudMoo-75}]\label{thm:qdF}
Let $\{\Gamma^{(k)}\}_{k\geq0}$ be an infinite sequence of bosonic $k$-particle density matrices on an arbitrary one-particle Hilbert space $\gH$, satisfying the consistency relations
\begin{equation}
\Gamma^{(k)}\geq0,\qquad \tr_{k+1}\Gamma^{(k+1)}=\Gamma^{(k)},\qquad \Gamma^{(0)}=1.
\label{eq:compatibility}
\end{equation}
Then there exists a probability measure $\mu$ on the unit sphere $S\gH=\{u\in \gH,\ \|u\|=1\}$, invariant under multiplication by phase factors, such that 
$$\Gamma^{(k)}=\int_{S\gH}|u^{\otimes k}\rangle\langle u^{\otimes k}|\,d\mu(u),\qquad\text{for all $k\geq1$.}$$
\end{theorem}

The result says that the infinite hierarchies of (appropriately normalized) bosonic density matrices can only be convex combinations of factorized states. In other words, the only extreme points of this convex set are the factorized states. This important theorem was already the basic tool used in the proof for confined systems in~\cite{FanSpoVer-80,VdBLewPul-88,RagWer-89,Werner-92}. It is very popular in quantum information theory~\cite{KonRen-05,FanVan-06,LevCer-09,ChrTon-09,RenCir-09,Chiribella-11,BraHar-12}.

The quantum de Finetti theorem makes the proof of Theorems~\ref{thm:energy} and~\ref{thm:states} very simple for confined systems. The main observation is that the energy can be written in terms of the two-particle density matrix as follows:
$$\frac{\pscal{\Psi_N,H_N\Psi_N}}{N}=\frac{1}{2}\tr \big(H_2\Gamma^{(2)}_{\Psi_N}\big).$$
However, the idea is to use the $\Gamma^{(k)}_{\Psi_N}$ for all $k\geq1$, to infer some information on the structure of their limits.
Since the density matrices are all normalized in the trace-class, hence are bounded, we may extract subsequences and assume that 
$\Gamma^{(k)}_{\Psi_{N_j}}\wto_\ast \Gamma^{(k)}$
weakly-$\ast$, for some $\Gamma^{(k)}$ and all $k\geq1$. For confined systems $H_2$ has a compact resolvent by assumption and the energy bounds can be used to show that no particle can escape to infinity. Hence $\Gamma^{(k)}_{\Psi_{N_j}}\to \Gamma^{(k)}$ strongly in the trace-class. This strong convergence implies that the limits $\Gamma^{(k)}$ satisfy the consistency relations~\eqref{eq:compatibility}. Using then Fatou's lemma and the quantum de Finetti Theorem~\ref{thm:qdF} for $\Gamma^{(2)}$, we infer that
\begin{align*}
\liminf_{N\to\ii}\frac{\pscal{\Psi_N,H_N\Psi_N}}{N}&=\liminf_{N\to\ii}\frac{1}{2}\tr \big(H_2\Gamma^{(2)}_{\Psi_N}\big)\\
&\geq \frac{1}{2}\tr \big(H_2\Gamma^{(2)}\big)
=\int_{SL^2(\R^d)}\frac{\pscal{u^{\otimes 2},H_2u^{\otimes 2}}}2\,d\mu(u)\\
&=\int_{SL^2(\R^d)}\cE(u)\,d\mu(u)\geq \eH\int_{SL^2(\R^d)}\,d\mu(u)=\eH.
\end{align*}
The upper bound implies that there is equality everywhere, hence that $E(N)/N\to\eH$ and that $\mu$ is supported on the set of minimizers for $\eH$. This ends the proof of Theorems~\ref{thm:energy} and~\ref{thm:states} for confined systems.

For unconfined systems, this simple argument breaks down, since the consistency relations~\eqref{eq:compatibility} do not pass to weak-$\ast$ limits in general. The method used in~\cite{LewNamRou-14} was to treat separately the particles that stay in a neighborhood of 0 (for which the quantum de Finetti theorem is valid) and those that escape. All the possible cases of $K$ particles escaping and $N-K$ staying, with $K$ of the order of $N$ have to be considered. These different events are handled using a method introduced in~\cite{Lewin-11}. For the particles that escape, our argument relies on the knowledge that their energy per particle converges to the translation-invariant GP energy. Since the potentials $A$ and $V$ vanish at infinity, the resulting one-particle operator $h=-\Delta$ is now real and positive-preserving and the arguments of Section~\ref{sec:energy} apply. In~\cite[Sec.~4.3]{LewNamRou-14} we designed a more involved proof that can even deal with arbitrary translation-invariant systems.

Our method is general and could be used in other situations. A byproduct of our approach is a generalization of the quantum de Finetti theorem to the case of weak convergence of density matrices.

\begin{theorem}[weak quantum de Finetti~\cite{LewNamRou-14}]\label{thm:wqdF}
Let $\Gamma_N$ be a sequence of bosonic $N$-particle state on $\bigotimes_s^N\gH$, and assume that the corresponding density matrices $\Gamma_N^{(k)}\wto_\ast \Gamma^{(k)}$ weakly-$\ast$ in the trace-class, for every $k\geq1$. 
Then there exists a probability density $\mu$ on the unit ball $B\gH=\{u\in \gH,\ \|u\|\leq 1\}$, invariant under multiplication by phase factors, such that 
$$\Gamma^{(k)}=\int_{B\gH}|u^{\otimes k}\rangle\langle u^{\otimes k}|\,d\mu(u),\qquad\text{for all $k\geq1$.}$$
\end{theorem}

That the final measure lives over the ball $B\gH$ instead of the unit sphere $S\gH$ is not surprising in the case of weak limits. It is important that it is always a probability measure. If all the particles are lost in the system, then $\mu$ is a Dirac delta at $u=0$. In the case of strong convergence, the limit $\Gamma^{(k)}$ must have a trace one, hence the measure $\mu$ has to live over the unit sphere $S\gH$ and one recovers Theorem~\ref{thm:qdF}. Theorem~\ref{thm:wqdF} has been shown to be a very useful tool to study the time-dependent BBGKY hierarchy~\cite{CheHaiPavSei-15} and to simplify arguments for the dilute limit~\cite{NamRouSei-15}. 

A result similar to Theorem~\ref{thm:wqdF}, stated in Fock space, appeared before in works~\cite{AmmNie-08,AmmNie-11} by Ammari and Nier. These authors developed a semi-classical theory in infinite dimension, calling $\mu$ a \emph{Wigner measure}, and this is really what the quantum de Finetti is about. 
The link with semi-classical analysis is better understood in a finite-dimensional space, $\dim(\gH)<\ii$. Indeed, in the bosonic $N$-particle space one has the \emph{resolution of the identity} in terms of factorized states
\begin{equation}
 \1_{\bigotimes_s^N\gH}=c_N \int_{S\gH}|u^{\otimes N}\rangle\langle u^{\otimes N}|\,du,\quad \text{with } c_N={{N+\dim(\gH)-1}\choose{\dim(\gH)-1}},
\label{eq:Schur}
 \end{equation}
which is similar to a coherent state representation. The formula follows from Schur's lemma, since the right side commutes with all the unitaries $U^{\otimes N}$ and the set spanned by these operators has a trivial commutant.

The formula~\eqref{eq:Schur} implies that the factorized states $u^{\otimes N}$ span the whole bosonic space $\bigotimes_s^N\gH$, but this is of course not an orthonormal basis. However, in the limit $N\to\ii$ its elements become more and more orthogonal since $\pscal{u^{\otimes N},v^{\otimes N}}=\pscal{u,v}^N\to0$ when $u$ is not parallel to $v$. This is similar to coherent states in the limit $\hbar\to0$.

Inspired by semi-classical analysis, it is then natural to compare an $N$-body state with the one built using the \emph{Husimi measure} associated with the coherent state representation~\eqref{eq:Schur}. This leads to the following quantitative version of the quantum de Finetti theorem, in a finite-dimensional space.

\begin{theorem}[Quantitative quantum de Finetti in finite-dimension~\cite{LewNamRou-15b}]
Let $\Gamma_N$ be any bosonic state over $\bigotimes_s^N\gH$, with $\dim\gH<\ii$, and define the corresponding Husimi probability measure by $d\mu_N(u)=c_N\langle u^{\otimes N},\Gamma_N u^{\otimes N}\rangle\,du$ on $S\gH$. Then, for every $1\leq k\leq N/(2\dim(\gH))$, we have in trace norm 
\begin{equation}
\norm{\Gamma_N^{(k)}-\int_{S\gH}|u^{\otimes k}\rangle\langle u^{\otimes k}|\,d\mu_N(u)}_1\leq \frac{2k\dim\gH}{N-k\dim\gH}.
\label{eq:quantitative_dF}
\end{equation}
\end{theorem}

This theorem is a slight improvement of a similar result in~\cite{ChrKonMitRen-07}. Related estimates have also appeared in~\cite{KonRen-05,FanVan-06,LevCer-09,ChrTon-09,RenCir-09,Chiribella-11,BraHar-12}.
In~\cite{LewNamRou-14} we prove the weak quantum de Finetti Theorem~\ref{thm:wqdF} by localizing the particles to a finite-dimensional space, where we could use~\eqref{eq:quantitative_dF}. Only after we have taken the limit $N\to\ii$ we removed the finite-dimensional localization. The fact that the measure might in the end live over the unit ball comes from the geometric localization procedure~\cite{Lewin-11}.

\section{The Bogoliubov excitation spectrum}\label{sec:Bogoliubov}

We now turn to the next order in the expansion of $H_N$, that is, the excitation spectrum and Bogoliubov's theory. From now on we assume that the GP energy has a unique ground state and that no particle escape, that is, 
$\cM=\{u_0\}$ with $\int_{\R^d}|u_0|^2=1.$
Of course, $u_0$ solves the GP equation which we rewrite as
$$h_0\,u_0=0,\qquad\text{with}\quad h_0:=h+|u_0|^2\ast w-\epsilon_0.$$
We also need the property that 
\begin{equation}
\int_{\R^d}\int_{\R^d}w(x-y)^2|u_0(x)|^2|u_0(y)|^2\,dx\,dy<\ii.
\label{eq:condition_K}
\end{equation}
This is not automatic since we have made no assumption on $w^2$, but this is true in most physical situations.
In order to be able to design a perturbation argument in a neighborhood of $u_0$, we also assume that the Hessian of $\cE$ is non-degenerate at this point. The Hessian is
\begin{align}
\frac12{\rm Hess}\;\cE_{\rm H}(u_0)(v,v)&=\pscal{v,h_0v}+\frac12\int_{\Omega}\int_{\Omega} w(x-y)u_0(x)u_0(y)\Big(\overline{v(x)}v(y)\nonumber\\
&\qquad\qquad+v(x)\overline{v(y)}+\overline{v(x)}\overline{v(y)}+ v(x)v(y)\Big)\,dx\,dy\nonumber\\
&=\frac12\pscal{\begin{pmatrix}v\\ \overline{v}\end{pmatrix},\begin{pmatrix}
h_0+K_1&K_2^*\\ K_2&\overline{h_0+K_1}\end{pmatrix}\begin{pmatrix}v\\ \overline{v}\end{pmatrix}}
\label{eq:Hessian}
\end{align}
for every $v$ orthogonal to $u_0$. Here $K_1$ and $K_2$ are the restrictions to $\{u_0\}^\perp\otimes_s\{u_0\}^\perp$ of the operators with kernel $w(x-y)u_0(x)\overline{u_0(y)}$ and $w(x-y)u_0(x)u_0(y)$, respectively. These are Hilbert-Schmidt (hence bounded) operators under the assumption~\eqref{eq:condition_K}. The condition that $u_0$ is non-degenerate means that 
\begin{equation}
\pscal{\begin{pmatrix}v\\ \overline{v}\end{pmatrix},\begin{pmatrix}
h_0+K_1&K_2^*\\ K_2&\overline{h_0+K_1}\end{pmatrix}\begin{pmatrix}v\\ \overline{v}\end{pmatrix}}\geq \eta\norm{v}^2
\label{eq:condition_Hess}
\end{equation}
for some $\eta>0$ and all $v\in\{u_0\}^\perp$. By~\cite[Appendix~A]{LewNamSerSol-15}, the condition~\eqref{eq:condition_Hess} is exactly what is needed to make sure that the second-quantized operator
\begin{multline}
\bH_{0}=\int a^\dagger(x)\big(h_0a\big)(x)\,dx+\int\int u_0(x)\overline{u_0(y)}w(x-y)a^\dagger(x)a(y)\,dx\,dy\\
+\frac{1}2\int\int u_0(x)u_0(y)w(x-y)a^\dagger(x)a^\dagger(y)\,dx\,dy+c.c. 
\label{eq:Bogoliubov2}
\end{multline}
is a well-defined bounded-below Hamiltonian on the Fock space with the mode $u_0$ removed,
$$\cF_+=\C\oplus\{u_0\}^\perp\oplus\bigoplus_{n\geq2}\bigotimes_s^n\{u_0\}^\perp.$$

Our goal is to prove that $\bH_0$ furnishes the excitation spectrum of the mean-field Hamiltonian $H_N$, that is, the excited states above $E(N)$. However, the two problems are posed in different Hilbert spaces. In~\cite{LewNamSerSol-15}, we have introduced a method to describe excitations of the condensate $u_0$, that makes the occurrence of $\cF_+$ very natural. The starting point is the remark that any $\Psi_N$ of the bosonic $N$-particle space can be written in the form
\begin{equation}
\Psi_N:=\phi_0\, u_0^{\otimes N}+ u_0^{\otimes (N-1)}\otimes_s\phi_1 + u_0^{\otimes (N-2)}\otimes_s\phi_2+\cdots + \phi_N
\label{eq:excitations}
\end{equation}
where $\phi_0\in\C$, $\phi_n\in\bigotimes_s^n\{u_0\}^\perp$ and $u_0$ is our GP minimizer (but so far it could be any fixed reference function). 
Furthermore, it is clear that the terms in~\eqref{eq:excitations} are all orthogonal, hence
$\norm{\Psi_N}^2=\sum_{n=0}^N\norm{\phi_n}^2.$
In other words, there is a natural isometry 
$$\begin{array}{cccl}
U_N: & \bigotimes_s^NL^2(\R^d) & \to & \displaystyle\cF_+^{\leq N}=\C\oplus\{u_0\}^\perp\oplus\bigoplus_{n=2}^N\bigotimes_s^n\{u_0\}^\perp\\[0.3cm]
 & \Psi_N & \mapsto & \phi_0\oplus \phi_1 \oplus\cdots \oplus \phi_N
  \end{array}
$$
from the $N$-particle space onto the truncated Fock space $\cF_+^{\leq N}$, which is itself a subspace of the full Fock space $\cF_+$. The unitary $U_N$ is adapted to the description of the fluctuations around $u_0$ and it plays the same role as the Weyl unitary for coherent states. After applying the unitary $U_N$ (which does not change the spectrum of $H_N$), we can settle the eigenvalue problem for $H_N$ in the truncated Fock space $\cF_+^{\leq N}$. In the limit $N\to\ii$, we obtain this way a problem posed on the Fock space $\cF_+$, involving the Bogoliubov Hamiltonian.

\begin{theorem}[Validity of Bogoliubov's theory~\cite{LewNamSerSol-15}]\label{thm:Bogoliubov}
Under the previous assumptions, we have the following results.

\smallskip

\noindent $(i)$ \emph{(Weak convergence to $\bH_0$)}. With $U_N$ the previous unitary, we have 
\begin{equation}
U_N\big(H_N-N\,\eH\big)U_N^*\wto \bH_0\quad\text{weakly.}
\label{eq:weak-CV}
\end{equation}

\smallskip

\noindent $(ii)$ \emph{(Convergence of the excitation spectrum).} We have the convergence
\begin{equation}
\lim_{N\to \infty}\big( \lambda_j(H_N)- N\eH\big) =\lambda_j(\mathbb{H}_0)
\label{eq:CV_excitation_spectrum}
\end{equation}
for every fixed $j$. Here $\lambda_j$ denotes the $j$th eigenvalue counted with multiplicity, or the bottom of the essential spectrum in case there are less than $j$ eigenvalue below.

\smallskip

\noindent $(iii)$ \emph{(Convergence of the ground state)}. The lowest eigenvalue of $\bH_0$ is always simple, with corresponding ground state 
$\Phi=\{\phi_n\}_{n\geq0}$ in $\cF_+$ (defined up to a phase factor). Hence the lowest 
eigenvalue of $H_N$ is also simple for $N$ large enough, with ground state $\Psi_N$ and a uniform spectral gap. Furthermore (with a correct choice of phase for $\Psi_N$), 
\begin{equation}
U_N \Psi_N \to \Phi
\label{eq:CV_gd_state}
\end{equation}
strongly in the Fock space $\cF_+$. In the $N$-particle space, this means 
\begin{equation}
\lim_{N\to \infty} \norm{\Psi_N-\sum_{n=0}^N(u_0)^{\otimes N-n}\otimes_s\phi_n}=0.
\label{eq:CV_gd_state2}
\end{equation}

\smallskip

\noindent $(iv)$ \emph{(Convergence of excited states)}. If $\lambda_j(\mathbb{H})$ is below the bottom of the essential spectrum of $\bH_0$, then we have a similar convergence result for the corresponding eigenvectors, up to subsequences in case of degeneracies.
\end{theorem}

Bogoliubov's theory has been investigated in many works, including completely integrable 1D systems~\cite{Girardeau-60,LieLin-63,Lieb-63,CalMar-69,Calogero-71,Sutherland-71}, the ground state energy of one and two-component Bose gases~\cite{LieSol-01,LieSol-04,Solovej-06}, the Lee-Huang-Yang formula of dilute gases~\cite{LieYng-98, ErdSchYau-08,GiuSei-09,YauYin-09} and the weakly imperfect Bose gas in a series of works reviewed in~\cite{ZagBru-01}. Our result was stimulated by~\cite{Seiringer-11,GreSei-13} which were the first to address the question in the mean-field model. Our work was then generalized in~\cite{NamSei-15} to cover non-degenerate local minima. The thermodynamic problem was investigated in~\cite{CorDerZin-09,DerNap-14}.

In addition to providing the expected convergence~\eqref{eq:CV_excitation_spectrum} of the excitation spectrum, Bogoliubov's theory also predicts the exact behavior, in norm, of the $N$-particle wavefunctions $\Psi_N$. We said that $\Psi_N$ is in general not close to $u_0^{\otimes N}$, since a finite number of the particles can be excited outside of the condensate without changing the energy to leading order. For eigenvectors of $H_N$, our result~\eqref{eq:CV_gd_state2} gives the exact form of these excitations, which are given by the components $\phi_n\in\bigotimes_s^n\{u_0\}^\perp$ of the corresponding Bogoliubov eigenvector in Fock space.

\section{The time-dependent mean-field model}

Everything we have discussed so far dealt with the ground state or the low-lying excited states. In this section we make some remarks on the time-dependent problem. The idea is to \emph{assume} that the system is condensed in a state $u_0$ with some excitations $(\phi_{n,0})$ at the initial time, and to prove that it stays close to a state of this form, with a dynamic condensate function $u(t)$ solving the time-dependent GP equation and excitations $\{\phi_n(t)\}$ evolving with Bogoliubov's equation.

More precisely, we give ourselves an arbitrary $u_0\in H^1(\R^d)$ with $\int_{\R^d}|u_0|^2=1$ describing the condensate and a sequence of functions $\phi_{n,0}\in \bigotimes_s^n\{u_0\}^\perp$ such that $\sum_{n=0}^\ii \int_{\R^{dn}}|\phi_{n,0}|^2=1$. If $\Psi_N(t)$ is the solution of the many-body Schr\"odinger equation
\begin{equation}
 \label{eq:Schrodinger-dynamics}
\left\{ \begin{gathered}
 i\, \dot{\Psi}_{N}(t) = H_N \Psi_{N}(t),  \hfill \\
  \Psi_{N}(0)=\sum_{n=0}^N u_0^{\otimes (N-n)}\otimes_s\phi_{n,0},\hfill \\ 
\end{gathered}  \right.
\end{equation}
then we have proved in~\cite{LewNamSch-15} that for all times $t\ge 0$, 
\begin{equation}
\lim_{N\to\ii}\norm{\Psi_N(t)-\sum_{n=0}^N u(t)^{\otimes (N-n)}\otimes_s\phi_n(t)}=0,
\label{eq:fluctuations}
\end{equation}
where $u(t)$ solves the time-dependent nonlinear GP equation
\begin{equation}
\left\{ \begin{gathered}
 i\, \dot u(t) =  \big(-\Delta +|u(t)|^2*w -\epsilon(t)\big) u(t),  \hfill \\
  u(0)=u_{0},\hfill \\ 
\end{gathered}  \right.
\end{equation}
with $\epsilon(t):=\frac12\iint_{\R^d\times\R^d}|u(t,x)|^2|u(t,y)|^2w(x-y)\,dx\,dy$, and where $\Phi(t):=\{\phi_n(t)\}_{n\geq0}$ solves the linear Bogoliubov equation in Fock space 
\begin{equation}
 \label{eq:Bogoliubov-dynamics}
\left
\{ \begin{gathered}
  i\, \dot{\Phi}(t) = \bH(t)  \Phi(t), \hfill \\
  \Phi(0)=(\phi_{n,0})_{n\geq0}. \hfill \\ 
\end{gathered}  
\right. 
\end{equation}
Here 
\begin{align} \label{eq:Bogoliubov-Hamiltonian}
\bH(t)=&\int_{\R^d}a^\dagger(x) \big(-\Delta+|u(t)|^2\ast w-\epsilon(t) + K_1(t) \big)a(x)\,dx\\
&+\frac12\iint_{\R^d\times\R^d}\Big(K_2(t,x,y)a^\dagger(x)a^\dagger(y)+\overline{K_2(t,x,y)}a(x)a(y)\Big)dx\,dy\nn
\end{align}
is the Bogoliubov Hamiltonian $\bH(t)$ which, this time, is defined on the whole Fock space including the condensate mode $u(t)$. The operators $K_1$ and $K_2$ are restricted to $\{u(t)\}^\perp$ by means of the projection $Q(t)=1-|u(t)\rangle\langle u(t)|$, namely $K_1(t)=Q(t)\tilde K_1(t) Q(t)$ and $K_2(t)=Q(t)\tilde K_2(t) \overline{Q(t)}$ with $\tilde K_1(t,x,y)=u(x)\overline{u(y)}w(x-y)$ and $\tilde K_2(t,x,y)=u(x){u(y)}w(x-y)$.
The evolved Bogoliubov state $\Phi(t)$ can be proved to be orthogonal to $u(t)$ for all times, as expected.

There are many works on the time-dependent mean-field limit, in the canonical or grand canonical setting, and discussing Bogoliubov corrections or not. Most existing works dealing with the Bogoliubov corrections focus on the description of the fluctuations around a \emph{coherent state in Fock space}, see for example ~\cite{Hepp-74,GinVel-79,GriMacMar-10,Chen-12}. With the exception of~\cite{BenKirSch-13}, our work~\cite{LewNamSch-15} seems to be the only one dealing with fluctuations close to a Hartree state $u(t)^{\otimes N}$ in the $N$-body space. The derivation of the GP time-dependent equation in the dilute limit was shown in~\cite{ErdSchYau-10,Pickl-10}, but the Bogoliubov corrections have not been established so far in this case.

\section{The infinitely extended Bose gas}

We have considered (possibly only locally) trapped Bose systems with a particular emphasis on the mean-field regime where the interaction is multiplied by a factor $1/N$. Similar questions can be raised for an infinite Bose gas in the thermodynamic limit.
Since we do not have a small coupling constant $\lambda\sim 1/N$ anymore, the stability of the interaction must be assumed: 
\begin{equation}
\sum_{1\leq j<k\leq N}w(x_j-x_k)\geq -CN.
\label{eq:w_stable}
\end{equation}
For simplicity we always think of $w$ being a continuous function with sufficiently fast decay. 
The thermodynamic limit for the Bose gas is done by confining the system to a cube $\Omega_L$ of side length $L$ and then taking the limit $L\to\ii$ with $\rho=N/|\Omega_L|$ fixed (the choice of boundary conditions is unimportant). The energy per particle is known to converge:
$$e(\rho,w)=\lim_{\substack{N\to\ii\\ N/|\Omega_L|\to\rho}}\frac1{N} \inf\sigma_{\bigotimes_s^NL^2(\Omega_L)} \left(\sum_{j=1}^N-\Delta_{x_j}+\sum_{1\leq j<k\leq N}w(x_j-x_k)\right).$$
In a similar fashion, we can define the Gross-Pitaevskii energy of the infinite Bose gas by
\begin{multline*}
\eH(\rho,w)=\lim_{\substack{N\to\ii\\ N/|\Omega_L|\to\rho}}\frac1{N}\inf_{\int_{\Omega_L}|u|^2=N}\Bigg(\int_{\Omega_L}|\nabla u|^2\\+\frac12\int_{\Omega_L}\int_{\Omega_L}|u(x)|^2|u(y)|^2w(x-y)\,dx\,dy\Bigg). \end{multline*}
Of course, $e(\rho,w)\leq \eH(\rho,w)$ by the variational principle. 
Note that the GP energy satisfies the simple relation
\begin{equation}
\eH(\rho\lambda,w/\lambda)=\eH(\rho,w).
\label{eq:scaling_GP}
\end{equation}
Taking $u=\sqrt{N/|\Omega_L|}$ (for Neumann boundary conditions), we find that $\eH(\rho,w)\leq ({\rho}/2)\int_{\R^d}w$
but in general there is no equality. 
Integrating~\eqref{eq:w_stable} against $\prod_{j=1}^N\nu(x_j)$ and taking $N\to\ii$, we find that 
$\int_{\R^d}\int_{\R^d}\nu(x)\nu(y)w(x-y)\,dx\,dy\geq 0$ for all $\nu\geq0$, and this shows that $\eH(\rho,w)\geq0$. 

We ask when $e(\rho,w)$ is close to $\eH(\rho,w)$ and, as before, there are several possible regimes. Three typical situations have been studied:
\begin{description}
 \item[$\bullet$ mean-field limit] $\rho\to\ii$ with $w$ replaced by $w/\rho$;
 \item[$\bullet$ dilute limit] $\rho\to0$ with $w$ fixed;
 \item[$\bullet$ van der Waals \textnormal{or} Ka\v{c} limit] $\rho$ fixed and $w$ replaced by $\gamma^dw(\gamma x)$ with $\gamma\to0$.
\end{description}
When the problem is investigated at finite volume, the volume of the sample comes into play, which further complicates the analysis. 

We remark that, in the literature, the name \emph{mean-field model} is often used for the system in Fock space with Hamiltonian
$\dGamma(-\Delta)+\int w/(2|\Omega_L|)\cN^2$
where $\cN$ is the particle number operator~\cite{BerLewSme-84}. It has the advantage that the interaction commutes with the one-body part, but it is only appropriate to describe repulsive interactions in certain regimes.

\subsection{Mean-field limit}
Arguing exactly as in Section~\ref{sec:energy}, we can prove the 
\begin{theorem}[Estimate on the energy difference]\label{thm:Bose_gas}
Assume that $w$ is classically stable~\eqref{eq:w_stable}, with $\int_{\R^d}|w|+|\widehat{w}|<\ii$. Then we have
\begin{equation}
\eH(\rho,w)-\frac{1}{2(2\pi)^{d/2}} \int_{\R^d}|\widehat{w}|   \leq e(\rho,w)\leq \eH(\rho,w).
\label{eq:estimate_energy_difference}
\end{equation}
\end{theorem}



The estimate~\eqref{eq:estimate_energy_difference} is only interesting in a regime where $\int|\widehat{w}|\ll \eH(\rho,w)$. The simplest regime of this kind is the mean-field limit $\rho\to\ii$ with $w$ replaced by $w/\rho$, since $\eH(\rho,w/\rho)=\eH(1,w)$ by~\eqref{eq:scaling_GP}. Theorem~\ref{thm:Bose_gas} then becomes
\begin{corollary}[Mean-field limit of the interacting the Bose gas]
If $w$ is classically stable~\eqref{eq:w_stable}, with $\int_{\R^d}|w|+|\widehat{w}|<\ii$, then 
\begin{equation}
\boxed{\lim_{\rho\to\ii}e(\rho,w/\rho)=\eH(1,w).}
\label{eq:mean-field_BG}
\end{equation}
\end{corollary}

This is well-known for positive-definite interactions, but does not seem to have been noticed for general interactions. The constraint that $\int_{\R^d}|\widehat{w}|<\ii$ can be dropped without changing~\eqref{eq:mean-field_BG}, with a worse error term in~\eqref{eq:estimate_energy_difference}.

\subsubsection*{Positive-definite interactions}
The simplest case, which has been studied the most in the literature, is when $\widehat{w}\geq0$. Then 
$\eH(\rho,w)= ({\rho}/2)\,\int_{\R^d}w$
for all $\rho>0$~\cite{Lieb-63b,LewPulSme-84}, and~\eqref{eq:mean-field_BG} becomes
\begin{equation}
\lim_{\rho\to\ii}e(\rho,w/\rho)=\frac12\int_{\R^d}w.
\label{eq:limit_MF}
\end{equation}
Bogoliubov's theory predicts that the second order correction is
\begin{multline}
\lim_{\rho\to\ii}\rho\left(e(\rho,w/ \rho)-\frac{1}2\int_{\R^d}w\right)\\
=-\frac{1}{2(2\pi)^d}\int_{\R^d}\left(|k|^2+(2\pi)^{d/2}\widehat{w}(k)-|k|\sqrt{|k|^2+2(2\pi)^{d/2}\widehat{w}(k)}\right)\,dk
\label{eq:CV_Bogo_gas_GS}
\end{multline}
and that the joint energy-momentum spectrum converges to the formula
\begin{equation}
|k|\sqrt{|k|^2+2(2\pi)^{d/2}\widehat{w}(k)}.
\label{eq:CV_Bogo_gas_ES}
\end{equation}
The latter gives the phonon spectrum when $w\propto \delta$ and a phonon-maxon-roton spectrum for generic $w$'s, as displayed in Figure~\ref{fig:phonon-roton}. A proof of~\eqref{eq:CV_Bogo_gas_GS} can be obtained by following the approach of~\cite{GiuSei-09}. The derivation of the excitation spectrum~\eqref{eq:CV_Bogo_gas_ES} is a famous open problem that was studied in~\cite{CorDerZin-09,DerNap-14}. 
The time-dependent problem was studied in~\cite{DecFroPicPiz-14}.

\subsubsection*{General interactions}
When $\widehat{w}$ has no sign, we need a condition ensuring that the limit is non-trivial. In particular we expect that $\eH(1,w)>0$, and the correction $-\rho^{-1 }\int|\widehat{w}|$ is then a lower order term. A very natural situation is when the interaction $w$ is \emph{superstable}, which means that 
\begin{equation}
 \sum_{1\leq j<k\leq N} w(x_j-x_k)\geq \frac{\epsilon}{2|\Omega|}N^2-CN 
 \label{eq:superstable}
\end{equation}
for some $\epsilon>0$ and any $x_j$'s in a large-enough cube $\Omega$. 
For a superstable interaction, we have
$\eH(1,w)\geq \epsilon/2>0$, as expected.

For a partially attractive potential, translation-invariance could be broken, with a GP minimizer which is not constant but instead displays some periodicity. This is confirmed by a numerical simulation with the Lennard-Jones potential in one dimension in Figure~\ref{fig:1D}. Similar results have been observed for soft ball potentials in higher dimensions in~\cite{JosPomRic-07b,AftBlaJer-09,KunKat-12}. For $\rho\gg1$ this proves a form of symmetry breaking for the many-particle problem as well, but a precise description of this effect for the many-particle ground state is a delicate question.

\begin{figure}[h]
\centering
\includegraphics[width=10cm]{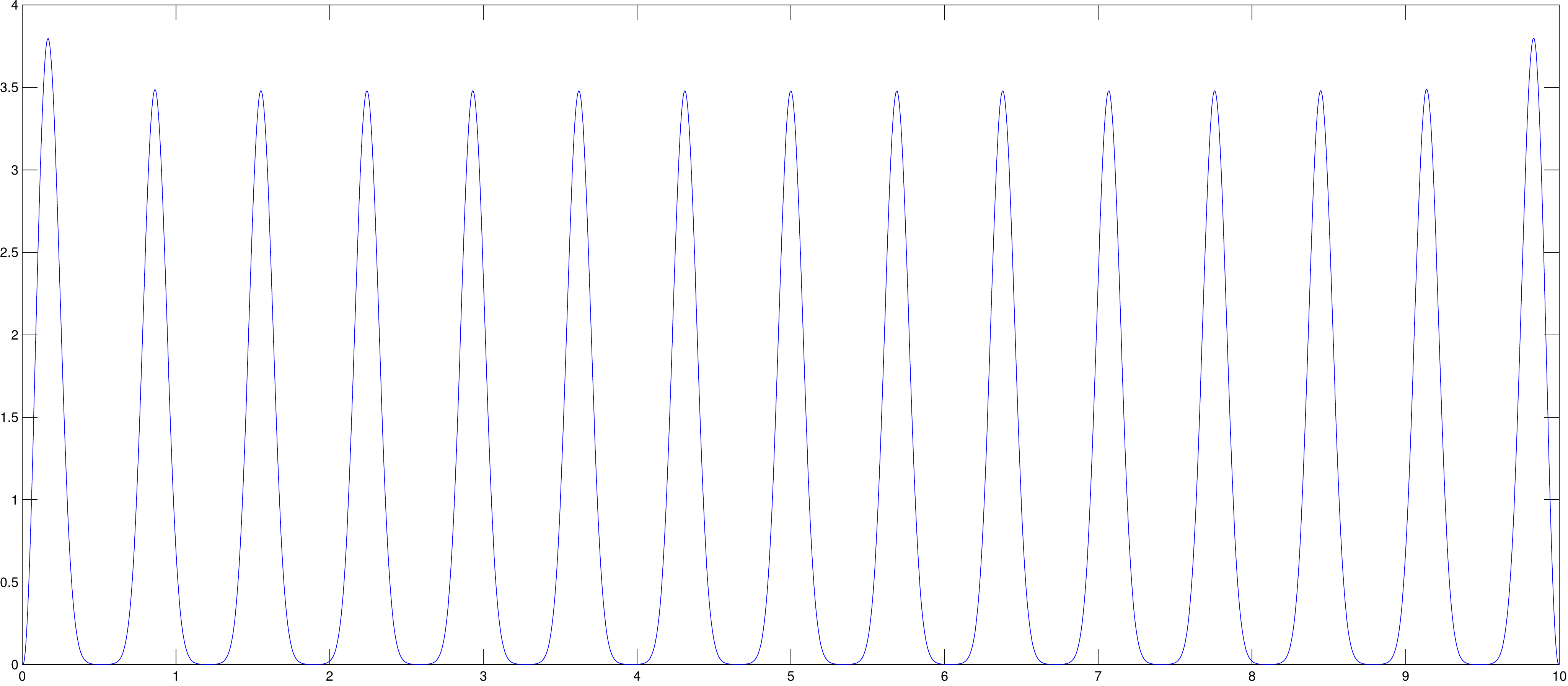}
\caption{\small Numerical calculation of the GP ground state density $|u|^2$ in 1D for the truncated Lennard-Jones potential $w(x)=\min(10^3,|x|^{-12}-|x|^{-6})$, showing the breaking of translational symmetry. Here $N=10$, $\rho=1$ and we have used Dirichlet boundary conditions.
\label{fig:1D}}
\end{figure}

\subsection{Dilute limit}
By scaling we have
$e(\rho,w)=\ell^{-2}\, e(\ell^d\rho,\ell^2w(\ell\cdot))$
In dimension $d=3$ it is natural to choose $\ell=\rho^{-1/2}$ which gives $e(\rho,w)=\rho\; e\left(\ell,w_\ell/\ell\right)$ with $w_\ell(x)=\ell^{3}w(\ell x)$. This way we can map the low density regime $\rho\to0$ (with fixed $w$) to an effective high density $\ell=\rho^{-1/2}\to\ii$ with the interaction $w_\ell$. For a positive-definite interaction, the asymptotics 
$$e(\rho,w)=4\pi\,a\,\rho+o(\rho)_{\rho\to0}$$
proved in~\cite{Dyson-57,LieYng-98,Seiringer-99,LieSeiSolYng-05}, can be obtained by formally using~\eqref{eq:limit_MF}, with $\int_{\R^3} w$ replaced by $8\pi a$. An important open problem is to derive Bogoliubov's correction (the Lee-Huang-Yang formula~\cite{LeeYan-57,LeeHuaYan-57,Lieb-63b})
$$e(\rho,w)=4\pi\,a\,\rho\left(1+\frac{128}{15\sqrt\pi}\sqrt{\rho a^3}
+o(\sqrt{\rho})_{\rho\to0}\right).$$
The upper bound was shown in~\cite{ErdSchYau-08,YauYin-09}.

\subsection{Van der Waals / Ka\v{c} limit}
The idea here is to spread the interaction very much to make it look like a constant, that is, $w$ is replaced by $w_\gamma=\gamma^d w(\gamma\cdot)$ with $\gamma\to0$ at fixed $\rho$. 
By scaling we see that $e(\rho,w_\gamma)=\gamma^2 e\big(\gamma^{-d}\rho,\gamma^{d-2}w\big)$. 
For a superstable interaction, we have $e(\rho,w_\gamma)\geq \epsilon \rho-C\gamma^d$ and~\eqref{eq:estimate_energy_difference} gives again 
$$\lim_{\gamma\to0}\frac{e(\rho,w_\gamma)}{\eH(\rho,w_\gamma)}=1.$$
Using~\eqref{eq:scaling_GP}, the GP energy becomes $\eH(\rho,w_\gamma)=\gamma^2 \eH\big(1,\rho\gamma^{-2}w\big)$ and this corresponds to adding a factor $\gamma^2$ in front of the kinetic energy. Therefore the limit is purely classical.
When $\widehat{w}\geq0$, we obtain in any dimension
$$\lim_{\gamma\to0}e(\rho,w_\gamma)=\frac{\rho}2\int_{\R^d} w.$$
This was studied in many works, including for instance~\cite{Lieb-66,BufSmePul-83,BerLewSme-83,SmeZag-87,MarPia-03,AlaPia-11}.

\subsection{Bose-Einstein condensation and critical temperature}
So far our discussion was focused on the ground state energy. A much more delicate matter is to understand Bose-Einstein condensation in the Bose gas, that is, whether the one-particle density matrix of the many-particle ground state has a macroscopic occupation in the GP ground state in the thermodynamic limit. 
This question is far from being settled at present. Some important works in this direction are~\cite{DysLieSim-78,LauVerZag-03}.

Another important problem is to understand the role of the temperature and in particular to find the critical temperature below which Bose-Einstein condensation occurs. Here we have only discussed the zero-temperature case. In the Ka\v{c} limit with positive-definite interactions, this was studied in~\cite{Lieb-66,BufSmePul-83,BerLewSme-83,SmeZag-87,MarPia-03,AlaPia-11}. For the dilute regime the situation is not yet well understood, see~\cite{Seiringer-08,SeiUel-09,Yin-10}.

\medskip

\small \noindent\textbf{Acknowledgement.} This research has received financial support from the European Research Council under the European Community's Seventh Framework Programme (FP7/2007-2013 Grant Agreement MNIQS 258023).


\end{document}